\documentclass[11pt]{article}
\usepackage{amsmath, amssymb, amsthm}
\usepackage{latexsym}
\usepackage{color}

\usepackage{mathtools}
\usepackage{amsfonts}
\usepackage{hyperref}

\newtheorem{thm}{Theorem}[section]

\newtheorem{lemma}[thm]{Lemma}

\newtheorem{claim}[thm]{Claim}

\newtheorem{theorem}[thm]{Theorem}
\newtheorem*{ntheorem}{Theorem}
\newtheorem{definition}[thm]{Definition}

\usepackage{fullpage}
\usepackage[capitalize]{cleveref}
\newcommand{\pnq}{\mathbb{F}_q P^n}
\newcommand{\pq}{\mathbb{F}_q P^2}

\newcommand{\pip}{\;|\;}
\newcommand{\closure}{\ensuremath{\mathsf{span}}}
\newcommand{\RR}{\ensuremath{\mathbb{R}}}
\newcommand{\NN}{\ensuremath{\mathbb{N}}}
\newcommand{\rank}{\mathsf{rank}}
\newcommand{\linspan}{\ensuremath{\mathsf{span}}}
\newcommand{\FF}{\mathbb{F}}
\newcommand{\suchthat}{\mathrel{:}}

\newcommand{\poly}{\operatorname{poly}}
\newcommand{\polylog}{\operatorname{polylog}}
\newcommand{\abs}[1]{\lvert#1\rvert}
\newcommand{\inner}[2]{#1\cdot #2}
\newcommand{\ps}[2]{\FF_{#1} P^{#2}}
\newcommand{\dm}{n}
\newcommand{\dimension}{\ensuremath{\mathsf{dim \:}}}

\def\final{1}  

\ifnum\final=0  
\newcommand{\lnote}[1]{[{\small Luis: \bf #1}]\marginpar{*}}
\newcommand{\nnote}[1]{[{\small Navin: \bf #1}]\marginpar{*}}
\newcommand{\snote}[1]{[{\small Santosh: \bf #1}]\marginpar{*}}
\newcommand{\sidecomment}[1]{\marginpar{\tiny #1}}
\newcommand{\details}[1]{[[#1]]}
\else 
\newcommand{\lnote}[1]{}
\newcommand{\nnote}[1]{}
\newcommand{\snote}[1]{}
\newcommand{\sidecomment}[1]{}
\newcommand{\details}[1]{}
\fi  
\newcommand{\email}[1]{\href{mailto:#1}{\texttt{#1}}}
\newcommand{\bigdotcup}{\ensuremath{\mathaccent\cdot{\bigcup}}}

\title{Query complexity of sampling and small geometric partitions}
\date{}

\author{Navin Goyal\\
Microsoft Research India \\
\email{navingo@microsoft.com}
\and
Luis Rademacher \\
Computer Science and Engineering \\
Ohio State University \\
\email{lrademac@cse.ohio-state.edu}
\and
Santosh Vempala \\
College of Computing \\
Georgia Institute of Technology \\
\email{vempala@cc.gatech.edu}}

\begin{document}

\maketitle
\begin{abstract}
In this paper we study the following problem:

\emph{Discrete partitioning problem (DPP):} Let $\pnq$ denote the $n$-dimensional finite projective space over
$\FF_q$. For positive integer $k \leq n$, let $\{ A^i\}_{i=1}^N$ be a partition of $(\pnq)^k$ such that
\begin{enumerate}
\item for all $i \leq N$, $A^i = \prod_{j=1}^k A^i_j$ (partition into product sets),
\item for all $i \leq N$, there is a $(k-1)$-dimensional subspace $L^i \subseteq \pnq$ such that $A^i \subseteq (L^i)^k$.
\end{enumerate}
What is the minimum value of $N$ as a function of $q,n,k$? We will be mainly interested in the case $k=n$.

DPP arises in an approach that we propose for proving lower bounds for the query complexity of generating random points from convex bodies. It is also related to other partitioning problems in combinatorics and complexity theory.
We conjecture an asymptotically optimal partition for DPP and show that it is optimal in two cases: When the dimension is low ($k=n=2$) and when the factors of the parts are structured, namely factors of a part are close to being a subspace. These structured partitions arise naturally as partitions induced by query
algorithms.
Our problem does not seem to be directly amenable to previous techniques for partitioning lower bounds such as rank
arguments, although rank arguments do lie at the core of our techniques.
\end{abstract}

\section{Introduction}\label{sec:intro}
In this paper we study the following problem:
\paragraph{Discrete partitioning problem (DPP).} Let $\pnq$ denote the $n$-dimensional finite projective space over
$\FF_q$ (see Sec.~\ref{sec:prelims} for a quick introduction to finite projective spaces and some related definitions). For positive integer $k \leq n$, let $\{ A^i\}_{i=1}^N$ be a partition of $(\pnq)^k$ such that
\begin{enumerate}
\item for all $i \leq N$, $A^i = \prod_{j=1}^k A^i_j$ (partition into product sets),
\item for all $i \leq N$, there is a $(k-1)$-dimensional subspace $L^i \subseteq \pnq$ such that $A^i \subseteq (L^i)^k$.
\end{enumerate}
What is the minimum value of $N$ as a function of $q,n,k$? We will be mainly interested in the case $k=n$.

DPP seems interesting in its own right and several related problems have been studied in the past; we 
discuss these later.
Before stating our results for DPP we discuss another motivation for studying it.
DPP arises in our approach for proving lower bounds for the query complexity of random sampling from convex bodies.
It is standard in this problem to give the convex body to an algorithm as a membership oracle, that is, a black box that when queried with a point in $\RR^n$ answers YES if the point is in the body and answers NO if the point is outside the body (e.g., \cite{GLS, lovasz2006fal}).
Given a convex body $K \in \RR^n$ via a membership oracle, by sampling from $K$ we mean
generating a random point from $K$ whose distribution is approximately uniform.
Sampling is one of the most useful primitives in the algorithmic theory of convex bodies
 (e.g., \cite{lovasz2006fal, DBLP:journals/jacm/BertsimasV04}). The currently best known algorithm \cite{lovasz2006fal} for sampling makes $O(n^4)$ membership oracle queries to generate one random point.
Improving this bound will directly improve the complexity of algorithms for volume computation and convex optimization.
On the other hand, the best known lower bound is just $\Omega(n)$. Thus, understanding the query complexity of sampling is an important problem.
 Notice that we are working with oracle algorithms, and so the lower bounds are on the query complexity and not on the computational complexity of sampling.


In this paper, we propose an approach for proving an $\Omega(n^2)$ lower bound on the query complexity of sampling.
The approach, discussed in Appendix~\ref{sec:reduction}, involves proving a lower bound on the number of queries for a problem that we call SPAN: Given $n-1$ vectors in $\RR^n$ via a natural oracle, find a hyperplane close to all of them. The problem SPAN reduces efficiently to sampling from convex bodies, so that a lower bound for SPAN implies a lower bound for sampling. Randomized oracle algorithms can be interpreted as distributions over deterministic decision trees. As is standard in lower bounds for randomized decision trees, it suffices to prove a lower bound on the size of a partition of the input space induced at the leaves of any small-depth deterministic decision tree with the following property: In most parts of this partition the value of the function being computed is nearly constant.
We call the problem of lower-bounding the size of this partition the \emph{continuous partitioning problem} (CPP).
While we do not solve CPP, we get insights into it by formulating an analogue of SPAN and its associated partitioning problem over finite fields and proving results in this setting. The rest of the introduction is devoted to discussion of these discrete problems. 



As the continuous problem SPAN only cares about the linear span of the input vectors, it is more conveniently stated not in a vector space but in the corresponding projective space, the space of all lines through the origin. The same can be said about the discrete analogue. Working over projective space makes counting arguments simpler.

\paragraph{Discrete span problem (DSPAN).} 
The input consists of $n$ points $v_1, \ldots, v_n \in \pnq$, where $\pnq$ is the finite projective space of dimension $n$ over the finite field $\FF_q$. 
The input can only be accessed via the following oracle: A query $x$ is an $(n-1)$-dimensional flat in $\pnq$; if $x$ contains all the points then the oracle answers YES, else it gives the least index $i$ such that $v_i$ does not lie in $x$. 
The problem is to find an $(n-1)$-flat (this is an abbreviation for $(n-1)$-dimensional subspace) containing all $v_i$'s.
The discrete SPAN problem is easily solved with $O(qn^2)$ queries using a deterministic algorithm.

We interpret algorithms for such a problem as randomized decision trees, namely a distribution on (deterministic) decision trees.
The leaves of such a deterministic decision tree induce a partition of $(\pnq)^n$, and the problem of lower-bounding the size of this partition is the discrete partitioning problem (DPP) stated at the beginning of this paper. 
The oracle described may seem a bit unnatural at first. 
It is motivated by the continuous problem and is chosen to be a mild strengthening of the ``membership oracle'' (which, in this case would just answer whether or not all $v_i$'s lie in $x$). A lower bound under the stronger oracle is also a valid lower bound for the weaker membership oracle because the algorithm can always ignore the additional information provided by the stronger oracle. At the same time, the strengthening adds the property that the parts of the induced partition are product sets 
(see \cref{sec:reduction} 
for more details).\footnote{In other words: 
For the DSPAN problem with membership oracle, parts are not necessarily product sets; 
when a membership query results in NO, we learn that some input vector $v_i$ is not in the queried hyperplane and the set of tuples consistent with this is not a product, it is actually the complement of a product. After $h$ queries, the part is the intersection of some product sets, resulting from YES queries, minus the union of some other product sets, resulting from NO queries. In particular, it is a product set minus the union of at most $h$ product sets. It is easy to show that any such set can be partitioned into $n^h$ product sets; the modified oracle is one way of showing this in our case.}
Lower bounds for partitions with product parts seem easier to prove than the general case and the product property is used crucially in our proofs.
Each such product is of the form $A_1 \times A_2 \times \dotsb \times A_n$, such that there is an $(n-1)$-flat $F \subset \pnq$ with $A_1 \times A_2 \times \dotsb \times A_n \subseteq F^n$. Moreover, each $A_i$ is somewhat structured: It can be represented as a flat minus a small number of other flats; thus each $A_i$ is close to a flat. See Lemma \ref{lem:flatminusflats} for a precise statement.

There are a few ways of formally relating DPP and DSPAN that we will sketch now.
A simple but weak way is to consider DSPAN and use Yao's minimax principle with cost giving the probability of failure of the randomized algorithm \cite[Section 2.2.2]{MRaghavan}, reducing the lower bound problem to proving a lower bound on the expected running time of a deterministic decision tree as in DSPAN (with the input drawn from a probability distribution) that is allowed to err with a small probability.
This leads to a variation of DPP with condition (2) relaxed so that each part is not necessarily \emph{fully} contained in the power of a $(k-1)$-flat, but only \emph{mostly} contained in such a flat.
In this paper we do not address this harder version of DPP.
A stronger connection is given by first observing that the correctness of a solution to DSPAN can be verified efficiently by querying the conjectured solution: The solution is correct iff the oracle answers YES.
Thus, the worst-case expected\footnote{Worst case over inputs of a given length, expected over the randomness of the algorithm.} running time of the best Las Vegas (i.e. always correct) algorithm is within a constant factor of the best Monte Carlo (i.e. correct with some probability) algorithm \cite[Exercise 1.3]{MRaghavan}.
That is, it is enough to prove a lower bound on the complexity of Las Vegas algorithms.
The use of Yao's minimax principle with cost equal to the running time reduces the lower bound for DSPAN problem to proving a lower bound on the average running time of deterministic algorithms against some input distribution (uniform, in our case), that is, a lower bound on the average depth of leafs (according to the input distribution).
For clarity we focus on the number of leafs in the main statement, but we actually prove that most (all but nearly a $1/q$ fraction) leafs are small (according to the input distribution), see Lemma \ref{lem:almostFlat} and the proof of Theorem \ref{thm:general} for a precise statement, as well as Section \ref{sec:techniques} for an overview of the argument.




Let us make some easy observations about DPP. The kind of partitions we are looking for always exist: Take any element
$(p_1, \ldots, p_k)$ in $(\pnq)^k$, where each $p_i$ is a point in $\pnq$. Together $p_1, \ldots, p_k$
span a $(k-1)$-flat. Thus the trivial partition in which each part is a singleton is a valid partition, giving an
upper bound on $N$ of size $((q^{n+1}-1)/(q-1))^k$, the total number of elements in $(\pnq)^k$. 
\details{Multivariate big-O implies there exists $M>0$ such that for $q,k,n \geq M$... So it is wrong to use this later with $k=n=2$. But $k=n$ is fine.}
For $q > k$, this is at most $eq^{kn}$.
\details{For $q > k$, it is at most $(q^n(1+1/q+1/q^2+\dotsb+1/q^n))^k \leq (q^n(1+1/q+1/q^2+\dotsb))^k = q^{kn} (1+1/(q-1))^k \leq e q^{kn}$.}

A lower bound
of $\Omega(q^{k(n-k+1)})$ (again assuming $q > k$) is obtained by a volume argument: The number of elements in $(\pnq)^k$,
as we noted, is $((q^{n+1}-1)/(q-1))^k$. The maximum number of elements in a part is 
$((q^{k}-1)/(q-1))^k$
;
this is because each factor is contained in a $(k-1)$-flat which has $(q^{k}-1)/(q-1)$
points. Thus $N$ must be at least $(\frac{q^{n+1}-1}{q-1})^k \cdot (\frac{q-1}{q^k-1})^k$. 
For $q>k$ this is at least $q^{kn}/eq^{k(k-1)})  \geq q^{k(n-k+1)}/e$.
Note that if we just wanted to cover instead of partition,
then $\Theta(q^{k(n-k+1)})$ is the tight upper and lower bound (when $q>k$): 
The covering given by the $k$th powers of all $(k-1)$-flats achieves the upper bound---it is well known that the number of $(k-1)$-flats in $\pnq$ is 
$\frac{(q^{n+1}-1)(q^{n+1}-q)\cdots (q^{n+1}-q^{k-1})}{(q^k-1)(q^k-q)\cdots (q^k-q^{k-1})} = eq^{k(n-k+1)}$.\details{Get $\leq((q^{n-k+2}-1)/(q-1))^k$ using $(a+t)/(b+t) \leq a/b$ whenever $a\geq b>0$ and $t\geq 0$. Then use previous paragraph.}

For the case $k = n$, the upper and lower bounds above become $O(q^{n^2})$ and $\Omega(q^{n})$.

\subsection{Related work}
Problems with similar flavor, namely finding a small partition of a product set into product sets with certain properties, abound in communication complexity, and are also studied in combinatorics.
Many techniques used to prove such lower bounds actually prove lower bounds on the covering number, with a few exceptions, such as the rank method~\cite{KushilevitzN} and certain lower bounds on the non-negative rank~\cite{1210.6970, DBLP:journals/eccc/BraunJLP13}; see also \cite{jainklauck} for some more recent work on partition lower bounds.
The covering problem is easy in our setting but the smallest covering seems to be much smaller than the smallest partition and thus does not provide insight into the size of the smallest partition. Our problem does not seem to be directly amenable to rank arguments or other techniques, although rank arguments do lie at the core of our techniques.
We now discuss some specific results related to our topic.

Alon et al.~\cite{AlonBHK02} consider the problem of partitioning a finite set
$A = A_1 \times \dotsb \times A_n$ (where $|A_i| \geq 2$ for all $i$) into parts of the form $B_1 \times \dotsb \times B_n$, where
$\emptyset \neq B_i \subsetneq A_i$ for $i = 1, \dotsc, n$.  They show that any such partition has size at least $2^n$. Our problem (DPP) is essentially a $q$-analogue of their problem.

Razborov~\cite{Razborov90} considers a more general partitioning problem in the context of formula complexity, albeit only for $k=2$. Briefly, suppose we have a covering of a set $U \times V = \cup_i \: C^i_1 \times C^i_2$. We say that a partition $\{A^i_1 \times A^i_2\}$  of $U \times V$ (so $\dot\cup_i\: A^i_1 \times A^i_2 = U \times V$, where $\dot\cup$ denotes disjoint union) is a refinement of the covering $\{C^i_1 \times C^i_2\}$ if for each part $A^i_1 \times A^i_2$ there is a $j$ such that $A^i_1 \times A^i_2 \subseteq C^j_1 \times C^j_2$.
Razborov considers the problem of proving a lower bound on the size of partitions refining certain coverings.
Clearly, our problem for $k=2$ is such a problem, as our partitions refine the covering of the $k$th powers of $(k-1)$-flats.
Razborov gives a method of proving lower bounds for the size of such partitions. This method seems to be specific to the $k=2$ case; for $k=2$, specialized to our problem, this method does not seem to give a bound better than $\Omega(q^2)$.

A lower bound for DPP would imply a lower bound for a deterministic \emph{number in hand} multiparty communication complexity problem (see \cite{KushilevitzN} for an account of communication complexity):
There are $k$ players; each player is given a private (unknown to other players) point from $\pnq$.
The players want to determine a $(k-1)$-flat containing the points of all the players. Notice that the output of the communication problem is not unique, and thus here we are interested in the communication complexity of a relation rather than that of a function.

Our problem fits into the
category of problems where one obtains a discrete model of a problem over the real field by changing the real field to a finite field. There are many examples of this interaction between the continuous and the discrete: The Kakeya problem over finite fields is one recent example with connections to the theory of computing; see, e.g., \cite{Dvir09}. Here also the problem becomes more tractable in the finite field setting.

\subsection{Our results}\label{sec:ourresults}

For $k=n=2$, the upper and lower bounds in Sec.~\ref{sec:intro} for the general problem become
$O(q^4)$ and $\Omega(q^2)$.
The truth turns out to be $\Theta(q^3)$:

\begin{theorem} \label{thm:product2}
In the discrete partitioning problem for $k=n=2$ the size of the smallest partition satisfies $N = \Theta(q^3)$.
\end{theorem}

For the general problem, we get an upper bound improving the trivial upper bound from
Sec.~\ref{sec:intro}, and generalizing the upper bound in Theorem~\ref{thm:product2}:
\begin{theorem}\label{thm:upper}
The discrete partitioning problem for $k=n$ and $q \geq 2n$ has a partition of size $q^{\binom{n+1}{2}}(1 + O(n/q))$.
\end{theorem}

In the previous theorem, the partition is made of parts whose factors are either a flat or a flat minus a lower dimensional flat, what we call an \emph{almost-flat}. For partitions of this kind we have a lower bound that matches our upper bound up to a multiplicative constant for $q \geq n$, and the constant approaches 1 for large $q$:
\begin{theorem}[partitioning lower bound for almost-flats]\label{thm:general}
For the discrete partitioning problem, if $k=n$ and each factor of every part is an almost-flat, then the partition size satisfies
\[
N \geq q^{\dm(\dm+1)/2} \left(1-\frac{1}{q} \left(\frac{q+1}{q-2}\right)^{\dm}\right) .
\]
\end{theorem}

Another motivation for studying such structured partitions comes from the fact that the partitions induced by decision trees for the DSPAN problem involve parts whose factors are flats minus a small number of flats. This is shown in Section \ref{sec:structure}. Our proof of Theorem~\ref{thm:general} does not seem to 
immediately generalize to this case.

Our approach for DPP, namely the idea of using the fraction of dependent tuples as a parameter of a part to lower bound the
size of the partition in Theorem~\ref{thm:general}, suggests using a similar idea for CPP, perhaps the density of ``approximately dependent'' tuples. While there remain technical difficulties in carrying out this approach in the continuous setting, it
appears promising, and is the direct result of considering DPP.

\subsection{Techniques} \label{sec:techniques}
In the proof of Theorem~\ref{thm:product2} the key idea is that the partitioning problem can be decomposed into smaller instances of simpler partitioning problems (Lemma~\ref{lem:decomposition}).
These smaller problems admit rank arguments for their lower bounds and are thus easy.
Our decomposition shows that on average each of these smaller problems requires a large partition via a rank argument, giving us a good overall bound.
While the rank lower bounds are fairly standard, the decomposition idea seems to be new.

The high level idea of the proof of Theorem~\ref{thm:general} is the following: We classify parts into two types, large and small (defined depending on the dimensions of its factors, later called ``non-dominated'' and ``dominated'' parts), where small parts contain at most about $q^{\dm^2/2}$ tuples each, while the total number of tuples is about $q^{\dm^2}$. On the other hand, each large part contains at least roughly a $1/q$ fraction of dependent tuples (meaning that their span has dimension less than $\dm-1$, Lemma \ref{lem:almostFlat}); while the set to be partitioned, $(\ps{q}{\dm})^\dm$, contains only about a $1/q^2$ fraction of dependent tuples, which implies that large parts can only cover about a $1/q$ fraction of all tuples. The rest must be covered by small parts, which by the previous discussion needs about $q^{\dm^2/2}$ parts (proof of Theorem \ref{thm:general}).  We remark that this high level idea has the flavor of the so-called corruption bound in
communication complexity (see \cite{DBLP:journals/cc/BeamePSW06}) and its
subsequent generalizations (e.g. \cite{ChakrabartiR10, jainklauck}).
Most of the work in our proof is in the lower bound for the fraction of dependent tuples in large parts (Lemma \ref{lem:almostFlat}), which is done by first partitioning any such part into parts having only 1-dimensional factors, and then handling this case by induction (Lemma \ref{lem:lines}) with the aid of a Sylvester-Gallai type property (Lemma \ref{lem:sylvester}).


\subsection{Organization}
The rest of the paper is organized as follows. Sec.~\ref{sec:prelims} contains relevant definitions.
Sec.~\ref{sec:2D} shows an optimal lower bound (up to constant factors) for DPP when $k=n=2$;
in Sec.~\ref{sec:construction} we present a non-trivial partition construction with structured parts;
the next section shows that this construction is essentially optimal for structured partitions.
Appendix~\ref{sec:reduction} gives more
details about how a solution to CPP would lead to a lower bound for sampling from convex bodies.

\section{Preliminaries}\label{sec:prelims}
For $n \in \NN$, let $[n] := \{1, \dotsc, n\}$.

We will work with projective spaces over finite fields. Projective spaces over finite fields are basic and extensively studied objects; see, e.g., \cite{BabaiFrankl} for an introduction.
Here we define projective spaces and note their relevant properties. In this paragraph, we follow the exposition of \cite{BabaiFrankl} closely.
Consider the $(n+1)$-dimensional linear space $W := \FF_q^{n+1}$ (where $\FF_q$ is the finite field of cardinality $q$ and $q$ is a prime power), and set $W^{\times}:= W \setminus \{0\}$.
Points in the $n$-dimensional projective space $\pnq$ over $\FF_q$ correspond to lines in $W$ through the origin.
More precisely, for $p \in W^{\times}$, consider the sets $\{a p \pip a \in \FF_q\setminus\{0\}\}$. Clearly, two such distinct sets are disjoint. These sets
together give a partition of $W^{\times}$. The projective space $\overline{W}$ consists of these sets as points. We define the dimension of $\overline{W}$ to be $n$ and denote this projective space by $\pnq$.
It is easy to see that $|\pnq| = (q^{n+1}-1)/(q-1)$; in particular, the cardinality of the \emph{projective plane} $\pq$ is $(q^3-1)/(q-1) = q^2 + q +1$. A \emph{flat} or \emph{subspace} of $\overline{W}$ is a set of the form $\overline{U}$ for a subspace $U$ of $W$. The dimension of $\overline{U}$ is defined to be $\dimension(U) -1$; thus $\dim(\emptyset)=-1$.
\details{The empty set is a flat of dimension $-1$.} 
We will often use the term $k$-flat for a $k$-dimensional flat.
For $S \subseteq \pnq$, denote by $\closure(S)$ the intersection of all flats containing $S$. 
\details{The interesection of two flats is a flat, so $\closure(S)$ is a flat.}
For a tuple $(p_1, \ldots, p_k)$ of $k$ points in $\pnq$, clearly $\dimension \closure \{p_1, \ldots, p_k\} \leq k-1$.
\details{To see this: Consider the corresponding lines in $\FF_q^{n+1}$. The maximum number of linearly independent vectors in their union is $k$, as one can pick at most one point from each line.} 
We say that  $(p_1, \ldots, p_k)$ is \emph{dependent} if $\dimension \closure \{p_1, \ldots, p_k\} < k-1$. 
Clearly, if a sub-tuple of a tuple is dependent then the whole tuple is dependent. 
A projective space of dimension 2 is called a projective plane, and flats of
dimension 1 are called (projective) lines. Projective planes have nice combinatorial properties; e.g., each point lies in
exactly $q+1$ lines, each line contains $q+1$ points, every pair of points lies on a unique line, and every pair of lines
intersects in a unique point. Higher dimensional spaces also have similar regularity properties.

\begin{definition}\label{def:almostflat}
We say that a subset of $\ps{q}{\dm}$ is an \emph{almost-flat} if it is either a flat or a $k$-flat minus a
flat of dimension at most $k-1$.
Let the \emph{dimension} of an almost-flat be the dimension of the minimal flat containing it. In particular, an almost-line
is a line or a line minus a point.
\end{definition}

We will need an appropriate counterpart for our setting (projective spaces over finite fields) of the familiar notion of 
orthogonal projection in projective spaces over the reals. This requires care because the notion of 
orthogonality can behave very differently over finite fields: In particular, a point can be orthogonal to itself. 

We define the projection using \emph{quotient by a flat}.  We will only use elementary properties of quotients and our discussion here is mostly 
self-contained. See, e.g., \cite{Faure} for a detailed treatment of quotients.
Let $F$ and $S$ be two flats in $\FF_q P^n$. An equivalence relation on $F\setminus S$ (an almost-flat) is given by $p \sim q$ iff
$\linspan((F \cap S) \cup \{p\}) = \linspan((F \cap S) \cup \{q\})$. The equivalence classes of $\sim$ are of the form
$\linspan((F \cap S) \cup \{p\})\setminus S =: [p]$ for $p \in F\setminus S$. 
The set of equivalence classes of $F\setminus S$ given by $\sim$ is called the quotient set and is denoted $F/S$.
Note that in our definition we did not require that $S \subseteq F$. 
Quotient set $F/S$ inherits the projective structure from $F$ in the natural way: For $p, q \in F\setminus S$ with $[p] \neq [q]$, 
the points are given by $[p]$, the lines are given by 
\[ \{[r] \suchthat r \in \linspan((F \cap S) \cup \{p\} \cup \{q\})\setminus S\}, \] 
and so on. 
Thus $F/S$ is a projective space of dimension $\dim(F)-\dim(F \cap S)-1$ living in $\FF_q P^{n-\dim(S)-1} = \FF_q P^n / S$. 
Notice that when $F \cap S = \emptyset$, then $\dim(F/S)=\dim(F)$, as $\dim(\emptyset)=-1$ according to our convention. 

For a flat $F' \subseteq F$, define $F'|_{F/S} := \{x \in F\setminus S \suchthat [x] \in [F']\}$, where $[F'] := \{[x] \suchthat x \in F'\setminus S\}$. 
In words, $F'|_{F/S}$ is the union of equivalence classes in $F\setminus S$ that intersect $F'$. 



We will use the following easy facts which we state without proof. 

Invariance of dependence under quotient:
\begin{claim} \label{claim:invariance-dependence}
Consider a tuple $t=(p_1, \dotsc, p_k)$,
$p_i \in \ps{q}{\dm}$, $p_1 \notin \{p_2, \dotsc, p_k\}$ and let $[p_2], \ldots, [p_k]$ be the images of $p_2, \ldots, p_k$ in the quotient
of the space by $p_1$. Then $t$ is dependent iff $([p_2], \dotsc, [p_k])$ is dependent.
\end{claim}
\details{Proof: ?}

Intersection of sub-flats with equivalence classes behaves nicely: 
\begin{claim} \label{claim:intersection-size}
For all equivalence classes $C \in F/S$ with non-empty intersection with a given flat $F'$, the intersection size $|C \cap F'|$ is the same. 
\end{claim}
\details{Proof: ?}

Dependence is a property of the equivalence classes: 
\begin{claim} \label{claim:dependence-prop-eq} 
Let $t = (p_1, \ldots, p_k, q_{k+1} , \ldots , q_{j} , \ldots q_m$), where the $p_i$'s and the $q_j$'s are 
points in $\pnq$. Let $t'$ be obtained from $t$ by replacing $q_j$ by $q'_j$. Also assume that $q_j, q'_j$ are in the same equivalence
class in the quotient of $\pnq$ by $S= \linspan(p_1, \ldots, p_k)$, i.e. $\linspan(S \cup \{q_i\}) = \linspan(S \cup \{q'_i\})$. Then 
either both $t$ and $t'$ are dependent or both are independent. 
\end{claim}
\details{Proof: ?}

\section{The discrete partitioning problem for $n=2$}\label{sec:2D}
In this section, instead of the projective space $\pnq$, we restrict ourselves to the
projective plane $\pq$.
Let us restate the problem for the projective plane. We want a partition of $(\pq)^2$ of the form
\begin{align} \label{def:partition2}
(\pq)^2 = \bigdotcup_{i=1}^{N} A^i_1 \times A^i_2,
\end{align}
such that for for all $i$ we have $A^i_1 \times A^i_2 \subseteq (L^i)^2$, where $L^i$
is a line in $\pq$.

We have $|(\pq)^2|=(q^2+q+1)^2 \approx q^4$. The upper and lower bounds
we discussed in Section~\ref{sec:intro} for the general problem now become $O(q^4)$ and $\Omega(q^2)$.
However, it turns out that $N=\Theta(q^3)$.

\paragraph{The upper bound.} First, note that for any point $p \in \pq$ there are
$q+1$ lines $L^p_1, L^p_2, \ldots, L^p_{q+1}$ through $p$. These lines only intersect in $p$ and together
they cover all of $\pq$.  Thus $L^p_1$ and $L^p_2 \setminus \{p\}, L^p_3 \setminus \{p\}, \ldots,
L^p_{q+1} \setminus \{p\}$ partition $\pq$. Now we can state our $O(q^3)$ size partition of $(\pq)^2$.
Each part is of the form $p \times L^p_1$ or $p \times (L^p_i \setminus \{p\})$ for
$i \in \{2, \ldots q+1\}$ and $p \in \pq$. Clearly these parts are mutually disjoint: For any two parts,
either the
first factors are different and disjoint, or if they are the same, then the second factors are disjoint by our
construction of the partition of $\pq$.  It is also clear that we cover all of $(\pq)^2$ in this way.
The size of this partition is $(q^2+q+1)(q+1) = O(q^3)$.

We now show that the above upper bound is the best possible up to a constant:
$N=\Omega(q^3)$.

\begin{ntheorem}[\ref{thm:product2} restated]
In the discrete partitioning problem for $k=2=n$ the partition size satisfies $N = \Theta(q^3)$.
\end{ntheorem}

\begin{proof} The key idea of the proof is that the partitioning problem can be decomposed into
smaller instances of simpler partitioning problems (Lemma~\ref{lem:decomposition} below). These smaller problems
admit rank arguments (similar to the one used in some proofs of a theorem by Graham and Pollak~\cite{GrahamP71}) for
their lower bounds. Our decomposition shows that on average each of these smaller
problems requires a large partition, giving us a good overall bound.

It will be useful to work without loss of generality with what we will call \emph{canonical} partitions, as it is easier to prove a lower bound for this restricted kind of partition.
We say that a partition of $(\pq)^2$ as in \eqref{def:partition2} is
canonical if each of its parts is canonical. We say that a part
$A_1 \times A_2$ is
canonical if either $A_1=A_2$ (\emph{square parts}) or $A_1 \cap A_2 = \emptyset$
(\emph{non-square parts}). In other words,
either the two factors are equal, or they are disjoint.

Given any partition $\{A^i_1 \times A^i_2\}$, we can construct a canonical
partition with at most $4$ times more number of parts as follows. For each part,
decompose it into four canonical parts:
\begin{align*}
A^i_1 \times A^i_2 &=
[(A^i_1 \cap A^i_2) \times (A^i_1 \cap A^i_2)] 
\dot\cup
[(A^i_1 \setminus A^i_2) \times (A^i_1 \cap A^i_2)] \\
&\qquad\dot\cup
[(A^i_1 \cap A^i_2) \times (A^i_2 \setminus A^i_1)] 
\dot\cup
[(A^i_1 \setminus A^i_2) \times (A^i_2 \setminus A^i_1)].
\end{align*}

Henceforth we assume that our partitions are canonical.

It will be helpful to think of
$(\pq)^2$ as a complete bipartite graph, with one copy of $\pq$ in the product representing
one side of vertices and the other copy representing the other side. Edges in this graph are then
the elements of $(\pq)^2$. Each canonical part can be thought of as an induced complete
bipartite subgraph.

Clearly, the number of square parts in any canonical partition is at most $q^2+q+1 = O(q^2)$.
We will show that the number of non-square parts is $\Omega(q^3)$.

Notice that if $\{S^i \times S^i \pip i \in [N]\}$ is the set of square parts, then
$\{S^i\}$ form a partition of $\pq$. Thus, $\{S^i \pip i \in [N]\}$ also induce a partition
of each line $L$; let $\phi(L)$ be the number of parts in such a partition of
$L$. Clearly $\phi(L) \leq q+1$.  The following lemma shows that on average
$\phi(L)$ is almost as large as $q+1$.

\begin{lemma} \label{lem:decomposition}
$\sum_{L} \phi(L) \geq q (q^2+q+1)$, where the summation is over all lines.
\end{lemma}
\begin{proof}
For any point $a$ there is some square part $S^i \times S^i$ such that
$a \in S^i$.  Now $a$ lies in $q+1$ lines, say, $L_1, \ldots, L_{q+1}$.  Since
our requirement on the partition is that $S^i$ should be completely in some
line, we have that for all but at most $1$ of the $q+1$ lines $L \in \{L_1,
\ldots, L_{q+1}\}$ we have $|L\cap S^i|=1$.  Thus $a$ appears as a singleton
in the partitions (induced by the square parts) for at least $q$ lines.  So each
of the $q^2+q+1$ points contributes at least $q$ to the sum, which gives the bound in the lemma.
\end{proof}

Remove the edges covered by square parts, then we are left with a bipartite graph whose edge set is partitioned by non-square parts.
In this graph, each line $L$ induces a bipartite subgraph $G(L)$ defined as follows: 
$G(L)$ is the bipartite subgraph induced by a copy of $L$ in the left vertices and a copy of $L$ in the right vertices. 
In other words, the edges of $G(L)$ are the edges in $L \times L$ not covered by square parts.
This implies that the edge set of each $G(L)$ is covered by non-square parts. 
Also, the edge sets of graphs $\{G(L)\}_L$ are disjoint by our construction. 
But a stronger property holds:
Each non-square part completely lies in one of the $G(L)$s.
More precisely, if $R^i_1 \times R^i_2$ is a non-square part such that
$R^i_1 \subseteq L$ and $R^i_2 \subseteq L$ for some line $L$, then
$(R^i_1 \times R^i_2) \cap (L'\times L') = \emptyset$ for all lines $L' \neq L$.

We know that
$G(L)$ looks like this: Let $L = S^1 \dot\cup \dotsb \dot\cup S^{\phi(L)}$ be the
partition of $L$ induced by square parts as above.  Then $G(L)$ has all the
edges in sets $S^i \times S^j$ for $i, j \in [\phi(L)], i \neq j$.  Now an easy
adaptation of the matrix proof of the Graham--Pollak theorem~\cite{GrahamP71} (see
Lemma~\ref{lem:GP} below)
gives that $G(L)$ needs $\phi(L)$ non-square parts. To see this,
choose one point $p^i$ from each $S^i$, and consider the subgraph of $G(L)$ induced by
the vertices in both color classes of $G(L)$ corresponding to points $\{p^1, \ldots, p^{\phi(L)}\}$.
Applying Lemma~\ref{lem:GP}
to this subgraph gives the required bound on the number of non-square parts.
Thus the total number of parts
we need is $\sum_L \phi(L) \geq q(q^2+q+1)$ by the lemma above.
\end{proof}

\medskip
We note that the proof did not make use of the algebraic structure of the projective plane,
and it holds for combinatorial projective planes as well.

\begin{lemma} \label{lem:GP}
Let $B = ((U,V),E)$ be a bipartite graph with $|U| = |V| = n$, and $E = \{(u_i, v_j) \pip
i, j \in [n] \text{ and } i \neq j\}$. (In other words, $B$ is a complete $n \times n$
bipartite graph minus a perfect matching.) Any partition of $E$ into complete bipartite graphs requires
at least $n$ graphs.
\end{lemma}
\begin{proof}
Consider the bipartite adjacency matrix $A(B)$ of $B$ (rows indexed by $U$ and columns by $V$,
and  $A(B)_{(u,v)} = 1$ if $(u,v) \in E$ else $A(B)_{(u,v)} = 0$). Let $B_1, \ldots, B_r$ be
complete bipartite subgraphs whose edges sets partition $E$. Then we can write
\begin{align}
A(B) =  \sum_{i \in [r]} A(B_i).
\end{align}
The algebra in the rest of the proof is over \RR. 
Now, notice that $\rank\:A(B) = n$ (this is because $A(B) = J-I$, where $J$ is the all ones matrix and $I$ is the identity matrix, after a suitable reordering of the vertices), but $\rank\:A(B_i)=1$ for $i \in [r]$. 
The subadditivity of rank implies that $r \geq n$.
\end{proof}


We remark that there are generalizations of the Graham--Pollak theorem for hypergraphs
\cite{Alon86, CioabaKV09} and it is natural to try to use these to solve the partitioning problem for
higher $k$. However, we have not succeeded in this.

\section{A small size partition} \label{sec:construction}
We construct a partition of $(\pnq)^n$ with size $O(q^{\binom{n+1}{2}})$, by
generalizing our partition construction for the product of two projective planes (Sec.~\ref{sec:2D}). More generally,
the same ideas give a partition of $(\pnq)^k$ with size $O(q^{\binom{k+1}{2}})$ (independent of $n$).
Informally, for the product of two projective planes
the parts were of type (point $\times$ almost-line). For $(\pnq)^n$, parts are of type (point
$\times$ almost-line $\times$ almost-2-flat $\times \dotsb \times$ almost-$(n-1)$-flat), where an
almost-$r$-flat is either an $r$-flat or an $r$-flat minus an $(r-1)$-subflat.
We now describe our construction in detail.

\medskip

\begin{proof}[Proof (of Theorem \ref{thm:upper})]
Let $1 \leq r < n$. For an $(r-1)$-flat $F$ consider $r$-flats $F_1, F_2, \ldots$ containing $F$.
There are $(q^{n+1}-q^r)/(q^{k+1}-q^r)$ 
such flats and any two of them intersect precisely in $F$. This provides a
partition of $\pnq$ into almost-$r$-flats with size $(q^{n+1}-q^r)(q^{k+1}-q^r)$ 
: The first part is $F_1$ and other
parts are $F_2 \setminus F, F_3 \setminus F, \ldots$ We call this partition a \emph{partition around $F$}.

Now to construct a partition of $(\pnq)^n$, it will be convenient to index the $n$ copies of $\pnq$ as
$P_1, \ldots, P_n$.  So we are considering a partition of $P_1 \times P_2 \times \dotsb \times P_n$.
We start by partitioning $P_1$. Let $\mathcal{P}_1$ be the partition of $P_1$ into
singletons. 
For each
$S_1 \in \mathcal{P}_1$, consider a partition of $P_2$ around $\closure(S_1) = S_1$. Denote this by
$\mathcal{P}_2(S_1)$. For $S_2 \in \mathcal{P}_2(S_1)$ consider partition of $P_3$ around
$\closure(S_2)$, and so on.

Our partition of $(\pnq)^n$ is then made up of all the parts of the form $S_1 \times \dotsb \times S_n$.
The number of choices for the first factor is $|\mathcal{P}_1| = (q^{n+1}-1)/(q-1)$. 
Having fixed the first factor $S_1$, the number of choices for the second
factor are $|\mathcal{P}_2(S_1)| = (q^{n+1}-q)/(q^2-q)$. 
And so on. So the total number of choices is
\begin{align*}
\frac{q^{n+1}-1}{q-1} \cdot \frac{q^{n+1}-q}{q^2-q} \cdot \frac{q^{n+1}-q^2}{q^3-q^2} \dotsm \frac{q^{n+1}-q^{n-1}}{q^{n}-q^{n-1}} 
&= \frac{q^{n+1}-1}{q-1} \cdot \frac{q^{n}-1}{q-1} \cdot \frac{q^{n-1}-1}{q-1} \dotsm \frac{q^{2}-1}{q-1} \\
&\leq \frac{q^{\binom{n+1}{2}}}{\bigl(1-\frac{1}{q}\bigr)^n} \\
&\leq \frac{q^{\binom{n+1}{2}}}{1-\frac{n}{q}}.
\end{align*}
For $q \geq 2n$ we have $1/(1-n/q) \leq 1+2n/q$. The claim follows.
\end{proof}

\section{The structure of decision trees for DSPAN}\label{sec:structure}
In this section we prove the claim from the introduction on the structure of the partition induced by a decision tree for the DSPAN problem: Each part is a product set, where each factor is a flat minus a few flats.
\begin{lemma}\label{lem:flatminusflats}
Consider a deterministic decision tree for DSPAN making at most $h$ queries\details{that is, of height $h$}. Let $A$ be a part of the partition of $(\pnq)^n$ induced by the leafs of the tree. Then
\begin{enumerate}
\item\label{item:thin}
There is an $(n-1)$-flat $F \subset \pnq$ with $A \subseteq F^n$.
\item\label{item:flatminusflats}
We can write $A = A_1 \times A_2 \times \dotsb \times A_n$ where each $A_i$ is of the form 
$G \setminus (G_1 \cup G_2 \cup \dotsb \cup G_h)$, where $G, G_1, \dotsc, G_h$ are flats.
\end{enumerate}
\end{lemma}
\begin{proof}
Part \ref{item:thin} must hold because the output of the tree is correct for DSPAN.

We prove the following strengthening of \ref{item:flatminusflats}: 
That \ref{item:flatminusflats} holds for the set of tuples $A$ associated to \emph{any} node of the decision tree with $h$ equal to the depth of the node. 
By induction on $h$.
It is clearly true for $h=0$ (no queries, the root) as in this case $A = (\pnq)^n$. 
For the inductive step, let $A \subseteq (\pnq)^n$ be the part associated to a node of depth $h$. 
By the inductive hypothesis, its parent part $A'$ is of the form $A_1 \times A_2 \times \dotsb \times A_n$, where each $A_i$ is of the form $G \setminus G_1 \cup G_2 \cup \dotsb \cup G_{h-1}$ where $G, G_1, \dotsc, G_{h-1}$ are (possibly empty) flats. 
The query that restricts $A'$ to get $A$ is some $(n-1)$-flat $p \subseteq \pnq$.
If the result of the query is YES, the interpretation of the query means that the restriction is: Intersect each $A_1, \dotsc, A_{n}$ with $p$.
If the result of the query is NO and index $i \in [n]$,
the interpretation of the query means that the restriction is: Intersect each $A_1, \dotsc, A_{i-1}$ with $p$, subtract $p$ from $A_i$ and leave $A_{i+1}, \dotsc, A_n$ unchanged. The claimed structure holds in both cases.
\end{proof}

\section{Lower bound for structured partitions}

In this section we show a lower bound for the discrete partitioning problem when factors of each part are almost-flats (Theorem~\ref{thm:general}, Def.~\ref{def:almostflat}). The outline of the proof in Sec.~\ref{sec:techniques} will be useful
for reading the proof below.


\begin{definition}[projective lines in general position]\label{def:generalPosition}
We say that a set of at most $\dm+1$ projective lines in $\ps{q}{\dm}$ is in \emph{general position} if for any $k \in [\dm-1]$ no $k+1$ of them are contained in a $k$-flat.
\end{definition}

\begin{lemma}[Sylvester-Gallai type property]\label{lem:sylvester}
Let $L$ be a set of at most $\dm+1$ projective lines in $\ps{q}{\dm}$ in general
position (Definition \ref{def:generalPosition}). Then there exists a projective line $l \in L$ that intersects the other projective lines in $L$ in at most 2 points, i.e. there are (at most) two points $p, q \in l$ such that $l \cap l' \in \{p, q\}$ for all 
$l' \in L\setminus\{l\}$. 
\end{lemma}
\begin{proof}
By induction on $\dm$. It is true for $\dm=1$. For general $\dm$, we will
define a sequence $l_1, l_2, \dotsc$ of lines in $L$. We will add
lines incrementally preserving the property that
$\dimension \linspan \{ l_1, \dotsc, l_i\} = i $. Start by picking any
line $l_1 \in L$. Pick a line $l_2 \in L\setminus \{l_1\}$ that intersects $l_1$
(if there is no such line then $l_1$ is the desired line).
In general, if there exists $l_i \in L \setminus \{l_1, \dotsc, l_{i-1}\}$
that intersects at least one of $l_1, \dotsc, l_{i-1}$, then we have
$\dimension \linspan \{l_1, \dotsc, l_i \} = \dimension \linspan \{l_1, \dotsc, l_{i-1} \} +1 = i$
($l_i$ cannot be contained in $\linspan \{l_1, \dotsc, l_{i-1} \}$
if $L$ is in general position). If no such $l_i$ exists, then
the inductive hypothesis applied to $\{l_1, \dotsc, l_{i-1}\}$ gives the line desired in the statement.
Suppose we pick all lines in $L$ in this way and the last line is $l_k$.
If $k < \dm + 1$, then $l_k$ intersects the others in one point. If $k=\dm + 1$, then
the fact that $L$ is in general position implies that $l_k$ intersects at most
one of $l_1, \dotsc, l_{\dm-1}$, and it can possibly intersect $l_\dm$. Thus, $l_k$
is the desired line.
\end{proof}
The previous lemma is tight in the following sense: For $\dm=2$, the case of the projective plane, any 3 lines in general position intersect pairwise.

\begin{lemma}[fraction of dependent tuples in products of almost-lines]\label{lem:lines}
Let $L= (l_i)_{i =1}^{\dm+1}$ be a family of almost-lines in $\ps{q}{\dm}$ with $q \geq 3$. Then the number of dependent tuples in $T=\prod_{i =1}^{\dm+1} l_i$ is at least $(q-2)^{\dm-1} (q-1)$.
\end{lemma}
\begin{proof}
By induction on $\dm$. For $n=1$, we are in the projective line of cardinality $q+1$, the two lines in $L$ coincide except for the missing points and the dependent tuples are pairs of equal points. Thus, there are at least $q-1$ dependent tuples.

For general $\dm$, if $L$ is not in general position, use the
inductive hypothesis on the subfamily of $k$ lines not in general position:
The number of dependent tuples in that subset is at least $(q-2)^{k-2} (q-1)$, any completion of such a dependent tuple to an $(\dm+1)$-tuple is also dependent and each can be completed in at least $q^{\dm+1-k}$ ways. Thus the number of dependent tuples in $L$ is at least $q^{\dm-k+1} (q-2)^{k-2} (q-1)$.

Otherwise, consider the line in $L$ given by Lemma~\ref{lem:sylvester} (applied to the completion of each almost-line to a line), say this line is $l_1$ and let $p$ be a point in this line that is not missing from it and such
that no other line in $L$ goes through it. 
Consider the quotient of the whole space by $p$. In the quotient, the image of a point $p' \neq p$ is $[p']$, and the image of a line $l$ not 
containing 
$p$ is the union of the images of the points in $l$. As the almost-lines in $(l_i)_{i=2}^{\dm+1}$ do not contain $p$, their images in the 
quotient are also 
almost-lines. Thus the inductive hypothesis can be used on the quotient space of dimension $\dm-1$ and the $\dm$ quotient lines to conclude 
that the product of the quotient lines contains at least $(q-2)^{\dm-2} (q-1)$ dependent tuples.

Now, by the invariance of dependence (Claim~\ref{claim:invariance-dependence}), we have that there are at least $(q-2)^{\dm-2} (q-1)$ 
dependent tuples in $T$ whose 
first coordinate is $p$. Also, there are at least $q-2$ choices of $p$, so there are at least $(q-2)^{\dm-1} (q-1)$ dependent tuples overall.
%
\end{proof}

\begin{definition}\label{def:dimensionpattern}
For $Q = \prod_{i=1}^k Q_i$, a product of subsets of $\ps{q}{\dm}$, where each $Q_i$ is an almost-flat, the \emph{dimension pattern of $Q$}, denoted $\dimension Q$ is the $k$-tuple of dimensions of the $Q_i$s sorted in non-decreasing order. We will consider the partial order on dimension patterns: $(s_1, \dotsc, s_k) \preceq (t_1, \dotsc, t_k)$ iff for all $i$ we have $s_i \leq t_i$.
\end{definition}


\begin{lemma}[dependence of non-dominated almost-flats]\label{lem:almostFlat}
Let $Q = \prod_{i=1}^\dm Q_i$ be a product of subsets of $\ps{q}{\dm-1}$, where each $Q_i$ is an almost-flat.
Assume 
\begin{equation}\label{equ:dim}
\dimension Q \npreceq (0, 1, \dotsc, \dm-1).
\end{equation}
Then the fraction of
dependent tuples in $Q$ is at least
\[
\frac{1}{q+1} \left(\frac{q-2}{q+1}\right)^{\dm-1}.
\]
\end{lemma}
\begin{proof}
The proof will reduce estimating the fraction in the general case to the case of lines, given by Lemma \ref{lem:lines}. We will do this by
first reducing to the case of partitions consisting of parts with \emph{minimal} dimension patterns satisfying \eqref{equ:dim} and then reducing to the case of product of lines.

The minimal dimension patterns satisfying \eqref{equ:dim} 
are the following $\dm-1$ patterns: $(1,1,\dotsc, 1)$, $(0,2, 2, \dotsc, 2)$, $(0,0,3,\dotsc, 3)$, $\dotsc$, $(0,\dotsc, 0,\dm-1,\dm-1)$. Formally, they are given by $(s_1, \dotsc, s_\dm)$ for $j=1, \dotsc,\dm-1$, where
\[
s_i =
\begin{cases}
j & i \geq j, \\
0 & i < j.
\end{cases}
\]

It suffices to prove the lemma for $Q$ with minimal dimension patterns satisfying \eqref{equ:dim}, because of the following two facts:
\begin{itemize}
\item 
A $Q$ with a non-minimal dimension pattern can be partitioned into parts with minimal dimension patterns. This is shown in the next claim.
\item
The fraction of dependent tuples in $Q$ is at least the minimum of such fractions for the parts in a partition of $Q$. 
\end{itemize}


\begin{claim}
Let $Q = \prod_{i=1}^\dm Q_i$ be a product of subsets of $\ps{q}{\dm-1}$, where each $Q_i$ is an almost-flat;
and let $Q$ satisfy \eqref{equ:dim}. Then $Q$ can be partitioned into parts with minimal dimension patterns and satisfying the assumptions of Lemma~\ref{lem:almostFlat}.
\end{claim}

\begin{proof}
Let $k+1$ be the least index $i$ such that $(\dimension Q)_i \geq i$; such an $i$ exists because of our assumption that \eqref{equ:dim} is satisfied. 
Then we claim that we can partition $Q$ into parts of the form
\begin{equation}\label{equ:minimal}
p_1 \times \dotsb \times p_k \times R_{k+1} \times \dotsb \times R_{\dm},
\end{equation}
where $p_i \in Q_i$, for $i \leq k$, are points, and $R_i \subseteq Q_i$ is an almost-flat of dimension $k+1$ for $i > k$.
We construct this partition by first partitioning individual factors in
$Q$, and then the resulting (refined) product partition of $Q$ will be our desired
partition.

Partitioning into flats of dimension 0 (points) is straightforward. For partitioning into higher dimensional parts
there are 3 cases depending on the factor being partitioned and the dimension of the target parts.
We will also assume that when we need to partition an almost flat it's of type $F_d \setminus F_{d''}$ with
$F_{d''} \subseteq F_d$. We have the following three cases. 

\begin{itemize}
\item In the first case, we want to partition a $d$-flat $F_d$ into almost-flats of dimension $d'$ for some $0 < d' \leq d$. Fix a $(d'-1)$-flat $F_{d'} \subseteq F_d$
arbitrarily, and consider the $d'$-dimensional flats $F(p) := \linspan (\{p\} \cup F_{d'})$ for
$p \in F_d \setminus F_{d'}$.  For two such points $p, p'$ we either have
$F(p) = F(p')$ or $F(p) \cap F(p') = F_{d'}$.  Thus we can construct a partition
of $F_d$ with one flat of the form $F(p)$ and almost-flats of the form
$F(p) \setminus F_{d'}$. More precisely, fix any point $p^* \in F_d \setminus F_{d'}$, then the partition is
\[
\{ F(p^*) \} \cup \{F(p) \setminus F_{d'} \suchthat p \in F_d \setminus F(p^*) \}.
\]

\item In the second case, we need to partition $F_d \setminus F_{d''}$, a $d$-flat minus a $d''$-flat, into almost-flats of dimension $d'$ for $d > d'' \geq d' > 0$. This is a slight modification of the previous argument: We fix a $(d'-1)$-flat $F_{d'} \subseteq F_{d''}$ arbitrarily and we can construct a partition
of $F_d \setminus F_{d''}$ with almost-flats of the form $F(p) \setminus F_{d'}$. The partition is
\[
\{F(p) \setminus F_{d'} \suchthat p \in F_d \setminus F_{d''} \}.
\]


\item In the third case, we need to partition $F_d \setminus F_{d''}$, a $d$-flat minus a $d''$-flat, into almost-flats of dimension $d'$ for $d > d' > d'' > 0$. This is again a slight modification of the previous argument: We arbitrarily fix a $d'-1$ dimensional flat $F_{d'} \subseteq F_{d}$ containing $F_{d''}$ and we construct a partition
of $F_d \setminus F_{d''}$ with one almost-flat of the form $F(p) \setminus F_{d''}$ and almost-flats of the form $F(p) \setminus F_{d'}$. More precisely, fix any point $p^* \in F_d \setminus F_{d'}$, the partition is
\[
\{ F(p^*) \setminus F_{d''} \} \cup \{F(p) \setminus F_{d'} \suchthat p \in F_d \setminus F(p^*) \}.
\]
\end{itemize}

Applying the above procedure to each factor $Q_i$ for $i>k$ with $d'=k$ we get the desired partition completing the proof of the claim.
\end{proof} 

To complete the proof of the lemma, we now reduce the case of minimal dimension patterns to the case of lines, which is handled by Lemma~\ref{lem:lines}. 
That lemma gives a lower bound for the fraction of dependent tuples for the product of $\dm+1$ lines in $\ps{q}{\dm}$. 
At this point in the proof we are dealing with parts as in \eqref{equ:minimal}, which have as factors $k$ points and $n-k$ almost-flats of dimension $k+1$. 
We could partition the almost-flats into lines to apply Lemma \ref{lem:lines} and ignore the first $k$ points of each tuple,
but then the lines would be living in $\ps{q}{\dm-1}$ with only $\dm-k$ lines and Lemma \ref{lem:lines} would not apply for $k\geq 1$. 
To fix this, we confine the almost-flats into a common $(\dm-k-1)$-dimensional space by ``projecting them orthogonal to $p_1, \dotsc, p_k$'', or more precisely, by taking the quotient by $S = \linspan(p_1, \ldots, p_k)$ and then appropriately modifying the $R_i$'s. We now describe this procedure. 

Let $Q$ be as in \eqref{equ:minimal}. If $(p_1, \dotsc, p_k)$ is dependent, there is nothing more to prove for this part as the fraction of dependent tuples is 1.

Otherwise, we sequentially go over $R_{k+1}, \ldots, R_{n}$ and replace them by $P_{k+1}, \dotsc,  P_{n}$ as described below. 
The new product set $Q' = p_1 \times \dotsb \times p_{k} \times P_{k+1} \times \dotsb \times P_n$ has the following properties: 
(1) $f(Q') \leq f(Q)$, where $f(Q)$ is defined to be the fraction of dependent tuples in $Q$; (2) each $P_i$ is an almost-flat; 
(3) each $P_i$ is the union of some of the equivalence classes induced by the quotient of $R_i$ by $S$. 
Thus, if we take the quotient, then each $P_i$ can be identified with an almost-flat living in a space isomorphic to $\ps{q}{\dm-k-1}$, and hence Lemma~\ref{lem:lines} is applicable after partitioning each $P_i$ into lines.

Now we explain the construction of the $P_i$'s which depends on two cases: 
(1) If $R_i$ is a flat, then set $P_i := R_i\setminus S$. 
(2) Suppose $R_i$ is an almost-flat,
i.e. $R_i = F_i \setminus F'_i$, where $F_i$ is a $(k+1)$-flat and $F'_i$ is a sub-flat of dimension at most $k$. 
Then set $P_i$ to either $F_i\setminus S$ or $(F_i\setminus S)\setminus (F'_i|_{F_i/S})$, 
whichever makes the current density of dependent tuples smaller. 
If the second option is empty, pick the first, which is never empty. (By the current density of dependent tuples we mean 
the density of dependent tuples in 
$p_1 \times \dotsb \times p_k \times P_{k+1} \times \dotsb \times P_i \times R_{i+1} \times \dotsb \times R_n$.)

We do not use the more natural choice of a straightforward quotient in the case of almost-flats (that is, $P_i = R_i/S$), as in that case the fraction of dependent tuples may increase or decrease. \details{Example?}
With our choice we will now show that the fraction of dependent tuples never increases. 

\begin{claim} \label{claim:surgery}
The fraction of dependent tuples in $Q$ is at least that in $Q'$.
\end{claim}
\begin{proof}
We will see the effect on the fraction of dependent tuples in each step of our procedure of replacing $R_i$ by $P_i$ as defined above. 
There will be several cases:
\begin{itemize}
\item If $R_i$ is a $(k+1)$-flat, then we set $P_i:= R_i \setminus S$. 
This cannot increase the fraction of dependent tuples because we removed $S$ and the tuples involving points of $S$ in the $i$'th position
are all dependent. 

\item If $R_i= F_i \setminus F'_i$, with $F_i \supset F'_i$, is an almost-flat, then we use a refinement of the previous argument. First, we replace $R_i$ by 
$R'_i = (F_i \setminus F'_i)\setminus S = (F_i \setminus S)\setminus (F'_i \setminus S)$; as before, this cannot increase the fraction of dependent 
tuples. Now we have two cases depending on the 
intersection pattern of $F'_i$ with the equivalence classes in $F_i/S$:
\begin{itemize}

\item $F_i/S=[F'_i]$, that is, $F'_i$ intersects all equivalence classes of the quotient $F_i/S$, and in this case each intersection is of the same cardinality by Claim~\ref{claim:intersection-size}. 
Therefore, by Claim~\ref{claim:dependence-prop-eq} the fraction of dependent tuples does not change when we replace $R'_i = (F_i\setminus S) \setminus (F'_i \setminus S)$ by $P_i = F_i\setminus S$.

\item $F_i/S \supsetneq [F'_i]$, that is, $F'_i$ does not intersect all equivalence classes of $F_i/S$. 
For $U \subseteq R'_i$, define $f(U)$ to be the fraction of dependent tuples in\footnote{We are overloading the function $f$ as it was
used with a different type of argument ($Q$) earlier, but this should not cause confusion.} 
\[
p_1 \times \dotsb \times p_k \times P_{k+1} \times \dotsb \times P_{i-1} \times U \times R_{i+1} \times \dotsb \times R_n.
\] 
Informally, we will either ``remove the equivalence classes intersected by $F'_i$'' or ``complete them'', whichever does not increase $f(\cdot)$.
More precisely, we will show that for one of the following choices of $P_i$ we have $f(P_i) \leq f(R'_i)$:  
set $P_i = (F_i\setminus S)\setminus (F'_i|_{F_i/S})$ (``remove''); or set $P_i = F_i\setminus S$ (``complete'').

\medskip
It remains to prove that one of these choices will not increase $f(\cdot)$. We need 
some notation. Denote the 
equivalence classes in $F_i/S$ by $C_1, \ldots, C_r, C_{r+1}, \ldots, C_{r+s}$. 
Let $c := |C_1| = |C_2| = \dotsb = |C_{r+s}|$.
Of these, $C_1, \ldots, C_r$ have nonempty intersection with $F'_i$. 
Let $c := |C_1| = |C_2| = \dotsb = |C_{r+s}|$.
By Claim~\ref{claim:intersection-size},
$|C_1 \cap F'_i| = \dotsb = |C_r \cap F'_i|$ and let us denote this common intersection size by $c'$. 
Let $\alpha$ be the fraction of dependent tuples induced by $C_1 \cup \dotsb \cup C_r = F'_i|_{F_i/S}$, and let $\beta$ be the 
fraction of dependent tuples induced by $C_{r+1} \cup \ldots \cup C_{r+s}$.
Then we have
\begin{align*}
f(R'_i) &= \frac{\alpha(c-c')r + \beta cs}{(c-c')r + cs},\\
f(F_i\setminus S) &= \frac{\alpha cr + \beta cs}{cr + cs},\\
f((F_i\setminus S)\setminus (F'_i|_{F_i/S})) &= \frac{\beta cs}{cs} = \beta.
\end{align*}
From the above expressions we see that if $\beta > \alpha$ then 
$f(F_i\setminus S) < f(R'_i)$; and if $\beta < \alpha$ then $f((F_i\setminus S)\setminus (F'_i|_{F_i/S})) < f(R'_i)$; and if $\alpha = \beta$ then either choice works. 
\end{itemize}

\end{itemize}
This completes the proof of the claim.
\end{proof}

Let $P = P_{k+1} \times \dotsb \times P_n$. 
By our construction, the fraction of dependent tuples in $Q'$ is no less than that in $P$. 
Define $P/S := (P_{k+1}/S) \times \dotsb \times (P_n/S)$, the result of taking the quotient with respect to $S$%
, where
$(P_j/S) \subseteq \FF_q P^{n-k-1}$ for $k < j \leq n$. Note that $P_j/S$ is an almost-flat. 
We have $f(P/S) = f(p_1 \times p_2 \times \dotsm \times p_k \times P)$. 
\details{is there a claim that implies the equality of both counts?}

Claim~\ref{claim:surgery} with the fact just noted implies that a lower bounding of the fraction of dependent tuples of $Q$ is given by a lower bound of the fraction of dependent tuples of a part having all factors of dimension 1 or more. 
Applying the partitioning argument from the first half of the proof once more to such a part, it is enough to lower bound the fraction of dependent tuples for a part having factors of dimension exactly
1 (minimal dimension pattern). 
The estimate in Lemma \ref{lem:lines} gives that each such part with $n-k$ factors has at least $(q-2)^{\dm-k-2} (q-1)$ dependent tuples. A part like that also has at most $(q+1)^{\dm-k}$ tuples and therefore a fraction of at least
\[
\frac{(q-2)^{\dm-k-2} (q-1) }{(q+1)^{\dm-k}}
\]
dependent tuples. As a function of $k$ only, this fraction is smallest when $k=0$, and thus it is at least
\[
\frac{(q-2)^{\dm-1} }{(q+1)^{\dm}}.
\]
We showed that this is a lower bound on the fraction of dependent tuples in $Q$. This completes the proof of the lemma.
\end{proof}

\begin{proof}[Proof (of Theorem \ref{thm:general})]
We will first estimate the fraction of dependent tuples in $(\ps{q}{\dm})^\dm$. Probabilistic language is helpful here. We consider a random tuple $T=(t_1, \dotsc, t_\dm)$ and we want an upper bound on the probability that it is dependent. Recall that the cardinality of an $i$-dimensional flat is $1+q + \dotsb + q^i$.
\begin{align*}
\Pr (T \text{ is dependent})
    &= \sum_{i=2}^\dm \Pr (\text{$(t_1, \dotsc, t_{i-1})$ is independent and $(t_1, \dotsc, t_{i})$ is dependent}) \\
    &\leq \sum_{i=2}^\dm \Pr (\text{$(t_1, \dotsc, t_{i})$ is dependent} \;|\; \text{$(t_1, \dotsc, t_{i-1})$ is independent}) \\
    &= \sum_{i=2}^\dm \frac{1+q+\dotsb + q^{i-2}}{1+q+\dotsb + q^\dm} \\
    &\leq \sum_{i=2}^\dm \frac{1}{q^{n-i+2}} \leq \sum_{i=2}^\infty \frac{1}{q^i} = \frac{1}{q(q-1)}.
\end{align*}
Thus the fraction of dependent tuples in $(\ps{q}{\dm})^\dm$ is at most $1/q(q-1)$.\footnote{This estimate is not too far from the true value: By picking the points in the tuple in sequence and considering the chance that the last point makes the tuple dependent (i.e. lies in a certain $(\dm-2)$-dimensional flat), we have that the fraction of dependent tuples is at least
\[
    \frac{q^{\dm-1}-1}{q^{\dm+1}-1} \geq \frac{1}{q^2} - \frac{1}{q^{\dm+1}}.
\]}
This and Lemma \ref{lem:almostFlat} imply that parts whose dimension pattern is not less than or equal to $(0,1, \dotsc, \dm-1)$ (non-dominated), can cover at most a
\[
\frac{1}{q(q-1)} \left( \frac{1}{q+1} \left(\frac{q-2}{q+1}\right)^{\dm-1} \right)^{-1}
\leq \frac{1}{q} \left(\frac{q+1}{q-2}\right)^{\dm}
\]
fraction of $(\ps{q}{\dm})^\dm$. The rest has to be covered with ``dominated'' parts, that is, parts whose dimension pattern is less than or equal to $(0,1, \dotsc, \dm-1)$. Any such part has cardinality at most $1 (q+1) \dotsm (q^{\dm-1} + \dotsb + 1)$. The total number of tuples to be covered by these parts is at least
\[
\left(1-\frac{1}{q} \left(\frac{q+1}{q-2}\right)^{\dm}\right) (q^\dm + \dotsb + 1)^\dm.
\]
This needs at least
\[
q^{\dm(\dm+1)/2} \left(1-\frac{1}{q} \left(\frac{q+1}{q-2}\right)^{\dm}\right)
\]
parts.
\end{proof}

\section{Acknowledgements}
We thank L\'aszl\'o Lov\'asz, Alexander Razborov, Michael Saks, Miklos Santha and David Xiao for useful discussions and pointers to the literature. We also thank an anonymous reviewer for careful reading and
useful comments.

\bibliography{sampling}
\bibliographystyle{abbrv}    

\appendix

\section{From sampling lower bound to the continuous partitioning problem}\label{sec:reduction}

A question that immediately arises when trying to prove a lower bound on sampling is that sampling is not a computational task in the usual sense of having a definite output.
A way to get around this problem is to prove a lower bounds for a problem that can be solved using sampling. An $\Omega(n)$ lower bound is easy: Consider the following set of $n$ bodies in $\RR^n$. 
For $i \in [n]$, define body $B_i = [0,1]^{i-1} \times [0,2] \times [0,1]^{n-i}$. 
In other words, $B_i$ is an axis-parallel cuboid with length $1$ along all but the $i$'th axis. 
Now consider a randomized algorithm that gets as input (via membership oracle) a uniformly randomly chosen body from the set of bodies just defined and its output is the index of the input body.
A straightforward application of Yao's min-max principle shows that any such algorithm must make $\Omega(n)$ membership queries to achieve a constant probability
of success. On the other hand, if sampling can be done with $q$ queries, then the body can be identified in $O(q)$
queries with constant probability of success: Suppose that the input body was $B_i$.  Sample a point by making $q$ queries.
With probability about $1/2$ its $i$'th coordinate is greater than $1$, thus telling us that the body is $B_i$.
We can improve the probability of success by repeating this. This gives $q = \Omega(n)$.

For a quadratic lower bound as a function of the dimension, our candidate hard algorithmic problem is the following. We are given a membership oracle for a convex body given by $\{x \in \RR^n \pip  \langle x, v_i \rangle \leq 1 \text{ for } i \in [n-1], \langle x, v \rangle \leq p(n)\}$, for $n-1$ unit vectors $v_1, \ldots, v_{n-1} \in S^{n-1}$ (the unit sphere in $\RR^n$) spanning a hyperplane, $v$ a normal to that hyperplane and $p(n)$ some fixed polynomial in $n$. The problem is to find $v$ approximately, or more precisely, a vector whose direction makes an angle with $v$ that is at most $1/\poly(n)$. As usual in algorithmic convexity, the oracle complexity of problems of this kind depends on the roundness of the input body \cite{GLS} and our problem as stated can have very high complexity as there is no a priori bound on the roundness of the input. For a meaningful worst-case lower bound for randomized algorithms one needs to restrict the input body so that it contains $r B_n$ and is contained in $R B_n$ for $R/r = \poly(n)$ (where $B_n$ is the unit ball in $\RR^n$). It's easy to show
algorithms solving the problem in this case with essentially quadratic number of queries. Yao's lemma implies that the probability of success of any randomized algorithm against the worst such input is at least the probability of success of the best deterministic algorithm against a distribution on inputs of our choice. Choosing $v_i$ uniformly and independently at random in $S^{n-1}$ and restricting the distribution to bodies satisfying the roundness condition is a natural option.  But it seems cleaner to just choose $v_i \in S^{n-1}$ uniformly at random without any additional constraint, prove a lower bound for deterministic algorithms against this distribution (say, an algorithm that fails with probability at most $p$, needs to make $q$ queries) and then argue that for a suitable choice of $r$ and $R$ the fraction of the distribution that is not well-rounded is at most $p/2$. So any algorithm when running on a distribution of well-rounded bodies needs to make at least $q$ queries to fail with probability at most $\frac{p/2}{1-(p/2)} = p/(2-p)$.
As before, it is easy to see that if we can sample with $O(q)$ queries then we can find a vector whose direction is within a $1/\poly(n)$ angle of $v$ in $O(q \polylog(n))$ queries with constant probability.


The next observation is that any deterministic algorithm against our distribution can be thought of as a decision tree (if we only care about the number of queries and not the computational complexity): Every node represents a query, the children of a node represent different choices depending on the result of a query and on leafs the algorithm stops and has to output a candidate vector. The leafs induce a partition of the support of the input distribution, which can be thought to be $(S^{n-1})^{n-1}$. The algorithm succeeds with high probability if for most parts, most tuples of $n-1$ vectors in the part have their normal direction near a fixed vector that depends on the part (``most'' here according to the input distribution). It simplifies the problem somewhat to assume that the oracle gives a bit more information than just YES or NO; instead, the modified oracle answers YES when the query point is in the body (as usual), but when the query point $x \in \RR^n$ is not in the body it answers NO and gives the least index among violated constraints (that is, $\min \{ i \suchthat \abs{\inner{x}{v_i}} > 1\}$). This idea was introduced in \cite{RV} and it has the following consequence (as shown there): The partition induced by the corresponding decision tree is made of product sets, namely, every part is of the form $P_1 \times \dotsb \times P_{n-1}$ where $P_i \subseteq S^{n-1}$. Clearly a lower bound on the number of queries for algorithms with the modified oracle is a valid lower bound for the original oracle.

For any given part that is a product set, it can be shown that if the angle of localization of the normals to its tuples is forced to be small enough (say, normal directions are within an angle $1/n^{O(1)}$ of a given direction for a $1-\alpha$ fraction of the part, for a small constant $\alpha$), then most of the part lies in a narrow ``band'', that is, it satisfies the following ``band condition'': For a set of the form $P_1 \times \dotsb \times P_{n-1} \subseteq (S^{n-1})^{n-1}$, there is a vector $v \in S^{n-1}$ such that $\mu (P_i \cap \{ x \suchthat \abs{\inner{v}{x}} \leq 1/n^{O(1)} \}) \geq (1-\alpha)\mu (P_i)$ for all $i$ (where $\mu$ denotes surface area).

The previous discussion reduces the problem of proving a lower bound $\Omega(n^2/\log n)$ for sampling to the following partitioning problem:

\paragraph{Continuous partitioning problem (informal).} Suppose that $Q_1, \dotsc, Q_k$ is a partition of $(S^{n-1})^{n-1}$ where each part is a product set and satisfies the above ``band condition''. A lower bound of 
$k \geq 2^{\Omega(n^2)}$ would translate to a quadratic query lower bound for the sampling problem. (The loss of a $\log$ factor is explained by the fact that the decision tree associated to the modified oracle has fan-out $\Theta(n)$).


A natural approach to solving the partitioning problem is to try to discretize the problem perhaps by subdividing the sphere into sufficiently small cells, and then working with these cells as atoms.
However, we found the discretization considered in this paper cleaner and more useful to work with.
Although we do not have a formal connection between the two problems they have very similar flavor and
insights from the
discrete version can be directly useful for the continuous version; for example, the partition 
in the proof of Theorem~\ref{thm:upper} translates into a non-trivial partition of $(S^{n-1})^{n-1}$ 
satisfying the band condition above. 
We now briefly describe the construction of this partition. 
We first give an infinite size 
partition which is essentially the one in the proof of Theorem~\ref{thm:upper} 
except that now we are working over the real field: The parts
are of the form $P_1 \times P_2 \times \dotsb \times P_{n-1}$. The first factor $P_1 \subset S^{n-1}$ 
is a point and its antipode in $S^{n-1}$ (this corresponds to a single point in the projective space). 
The second factor $P_2 \subset S^{n-1}$ is obtained from a great circle $C$ in $S^{n-1}$ containing 
$P_1$; factor $P_2$ is either $C$ itself or $C \setminus P_1$ (this corresponds to a line in the projective space). 
Factor $P_3$ is obtained from the intersection of $S^{n-1}$ with a 3-dimensional subspace 
of $\RR^{n}$
containing $P_2$ (this corresponds to a plane in the projective space), and so on. 
To turn this into a finite partition, we ``fatten'' each factor by $1/p'(n)$, 
where the polynomial $p'(n)$ is related to the precision with which we have to determine $v$, the normal
to $v_1, \ldots, v_{n-1}$. For points, this fattening is achieved by subdividing $S^{n-1}$ into 
regions of diameter at most $1/p'(n)$. For a given $P_1$, the second factor $P_2$ is obtained by similarly partitioning 
$S^{n-1}$ into a finite number of regions such that for each region there is a great circle with every
point in the region within distance $1/p'(n)$ from the great circle, and one of these regions 
contains $P_1$ and the others are disjoint from $P_1$. We proceed similarly for higher dimensional factors.


\end{document}